\documentclass{article}
\usepackage[T1]{fontenc}
\usepackage{amsmath}
\usepackage{amsthm}
\usepackage{amssymb}
\usepackage{setspace}
\usepackage[authoryear]{natbib}
\makeatletter
\newcommand{\lyxaddress}[1]{
	\par {\raggedright #1
	\vspace{1.4em}
	\noindent\par}
}
\newenvironment{lyxlist}[1]
	{\begin{list}{}
		{\settowidth{\labelwidth}{#1}
		 \setlength{\leftmargin}{\labelwidth}
		 \addtolength{\leftmargin}{\labelsep}
		 }}
	{\end{list}}
\theoremstyle{plain}
\newtheorem{thm}{\protect\theoremname}
\theoremstyle{plain}
\newtheorem{prop}[thm]{\protect\propositionname}
\ifx\proof\undefined
\newenvironment{proof}[1][\protect\proofname]{\par
	\normalfont\topsep6\p@\@plus6\p@\relax
	\trivlist
	\itemindent\parindent
	\item[\hskip\labelsep\scshape #1]\ignorespaces
}{%
	\endtrivlist\@endpefalse
}
\providecommand{\proofname}{Proof}
\fi
\theoremstyle{plain}
\newtheorem{cor}[thm]{\protect\corollaryname}

\date{}

\makeatother

\providecommand{\corollaryname}{Corollary}
\providecommand{\propositionname}{Proposition}
\providecommand{\theoremname}{Theorem}

\begin{document}

\title{On Axiomatization of Lewis' Conditional Logics}
\author{Xuefeng Wen}
\maketitle

\lyxaddress{\begin{center}
Institute of Logic and Cognition\\Department of Philosophy\\Sun Yat-sen University\\
wxflogic@gmail.com
\par\end{center}}
\begin{abstract}
This paper first shows that the popular axiomatic systems proposed
by Nute for Lewis' conditional logics are not equivalent to Lewis'
original systems. In particular, the axiom CA which is derivable in
Lewis' systems is not derivable in Nute's systems. Then the paper
proposes a new set of axiomatizations for Lewis' conditional logics,
without using CSO, or RCEA, or the rule of interchange of logical
equivalents. Instead, the new axiomatizations adopt two axioms which
correspond to cautious monotonicity and cautious cut in nonmonotonic
logics, respectively. Finally, the paper gives a simple resolution
to a puzzle about the controversial axiom of simplification of disjunctive
antecedents, using a long neglected axiom in one of Lewis' systems
for conditional logics.

\end{abstract}

\section{Introduction}

Lewis proposed two conditional logics, denoted by $\mathbf{V}$ and
$\mathbf{VC}$, respectively. Each of them has three different axiomatizations
in the literature. Two were proposed by Lewis himself, one in \citep{Lewis1971},
where $\mathbf{V}$ and $\mathbf{VC}$ were named $\mathbf{C0}$ and
$\mathbf{C1}$, respectively, the other in \citep{Lewis1973a}. A
third formulation was offered by \citet{Nute1980,Nute1984,Nute2001}.
Lewis' formulations have less but some cumbersome axioms. Nute's formulations
have more but neater axioms, making them easier to compare with other
systems. Thus, Nute's axiomatizations are more popular in the literature
now. When referring to Lewis' conditional logics, often are Nute's
axiomatizations presented, for instance in \citep{ArloCosta2014}
and \citep{Pozzato2010}.\footnote{In \citep{ArloCosta2014}, the author wrote: ``...it is useful to
see first that the system VC can be axiomatized via the axioms ID,
MP, MOD, CSO, CV and CS with RCEC and RCK as rules of inference.''} I will show in this paper, however, that Nute's systems are not equivalent
to Lewis' original ones. In particular, the axiom CA derivable from
Lewis' systems is not derivable from Nute's systems. By replacing
MOD with CA in Nute's systems, the defects can be amended. 

Both Lewis' systems in \citep{Lewis1971} and Nute's systems contain
the axiom CSO, which says that bi-conditionally implied propositions
can be interchanged with each other for antecedents. From CSO together
with RCE (namely a conditional from $\varphi$ to $\psi$ can be derived
if $\varphi$ entails $\psi$), the rule of interchange of logical
equivalents for antecedents (RCEA, henceforth) can be derived. Instead
of CSO, Lewis' systems in \citep{Lewis1973a} contain the rule of
interchange of logical equivalents (RE, henceforth). I will propose
some new axiomatizations for Lewis' logics. They contain neither CSO,
nor RCEA or RE, and hence may shed light on nonclassical conditional
logics, where these axiom and rules are invalidated. The new systems
I propose indicate that it is hard to abandon these axiom and rules
in conditional logics, since they can be recovered from other intuitive
axioms.

Finally, I will show that an axiom in one of Lewis' systems can be
used to solve a puzzle triggered by the controversial axiom of simplification
of disjunctive antecedents (SDA, henceforth), which is intuitively
valid but trivializes conditional implication to strict implication
if added to any conditional logic with RCEA.

\section{Preliminaries}

For reference, I list all related axioms and rules for conditional
logics in this paper as follows:
\begin{lyxlist}{00.00.0000}
\item [{(PC)}] All tautologies and derivable rules in classical logic
\item [{(ID)}] $\varphi>\varphi$
\item [{(CM)}] $(\varphi>\psi\land\chi)\to(\varphi>\psi)\land(\varphi>\chi)$
\item [{(CC)}] $(\varphi>\psi)\land(\varphi>\chi)\to(\varphi>\psi\land\chi)$
\item [{(CV)}] $(\varphi>\chi)\land\neg(\varphi>\neg\psi)\to(\varphi\land\psi>\chi)$
\item [{(CA)}] $(\varphi>\chi)\land(\psi>\chi)\to(\varphi\lor\psi>\chi)$
\item [{(AC)}] $(\varphi>\psi)\land(\varphi>\chi)\to(\varphi\land\psi>\chi)$
\item [{(RT)}] $(\varphi>\psi)\land(\psi\land\varphi>\chi)\to(\varphi>\chi)$
\item [{(CSO)}] $(\varphi>\psi)\land(\psi>\varphi)\to((\varphi>\chi)\leftrightarrow(\psi>\chi))$
\item [{(MOD)}] $(\varphi>\neg\varphi)\to(\psi>\neg\varphi)$
\item [{(DAE)}] $(\varphi\lor\psi>\varphi)\lor(\varphi\lor\psi>\psi)\lor((\varphi\lor\psi>\chi)\leftrightarrow(\varphi>\chi)\land(\psi>\chi))$
\item [{(PIE)}] $(\varphi>\neg\psi)\lor((\varphi\land\psi>\chi)\leftrightarrow(\varphi>(\psi\to\chi)))$
\item [{(CMP)}] $(\varphi>\psi)\to(\varphi\to\psi)$
\item [{(CS)}] $\varphi\land\psi\to(\varphi>\psi)$
\item [{(SDA)}] $(\varphi\lor\psi>\chi)\to(\varphi>\chi)\land(\psi>\chi)$
\item [{(RCM)}] $\dfrac{\varphi\to\psi}{(\chi>\varphi)\to(\chi>\psi)}$
\item [{(RCE)}] $\dfrac{\varphi\to\psi}{\varphi>\psi}$
\item [{(RCN)}] $\dfrac{\psi}{\varphi>\psi}$
\item [{(RCK)}] $\dfrac{\psi_{1}\land\ldots\land\psi_{n}\to\psi}{(\varphi>\psi_{1})\land\ldots\land(\varphi>\psi_{n})\to(\varphi>\psi)}$
$(n\geq0)$
\item [{(RCEA)}] $\dfrac{\varphi\leftrightarrow\psi}{(\varphi>\chi)\leftrightarrow(\psi>\chi)}$
\item [{(RCEC)}] $\dfrac{\varphi\leftrightarrow\psi}{(\chi>\varphi)\leftrightarrow(\chi>\psi)}$
\item [{(RE)}] $\dfrac{\psi\leftrightarrow\psi'}{\varphi\leftrightarrow\varphi[\psi/\psi']}$
\end{lyxlist}
All the axioms and rules above had been discussed in the literature
(e.g. \citealp{Lewis1973a,Nute1980,Nute1984}) before. Note that I
slightly reformulate the axiom MOD here. The standard formulation
of MOD in the literature (including Lewis' works) is
\begin{lyxlist}{00.00.0000}
\item [{MOD'}] $(\neg\varphi>\varphi)\to(\psi>\varphi)$.
\end{lyxlist}
There are two reasons why I reformulate it. One is that it is this
reformulation rather than the standard one that corresponds directly
to the associated model condition of worlds selection functions, normally
formulated in the literature as follows:
\begin{lyxlist}{00.00.0000}
\item [{(mod)}] $f(w,\varphi)=\emptyset\Longrightarrow f(w,\psi)\cap[\varphi]=\emptyset$,
\end{lyxlist}
where $[\psi]$ denotes the truth set of $\psi$, and $f$ is the
selection function, associating with a possible world $w$ and a sentence
$\varphi$ a set of $\varphi$- worlds that are closest to $w$. Rather,
the standard formulation MOD' corresponds to the following condition
instead:
\begin{lyxlist}{00.00.0000}
\item [{(mod')}] $f(w,\neg\varphi)=\emptyset\Longrightarrow f(w,\psi)\cap[\neg\varphi]=\emptyset$.
\end{lyxlist}
Of course, if the rule RCEA or RE is available, the difference between
the two formulations is immaterial. But if one works on conditional
logics without such rules admissible, the two formulations might turn
out to be very different. This is related to my second reason for
choosing the reformulation. In a proof of the derivation of CSO from
MOD' and PIE below, I find if the reformulation MOD is used then the
rule RCEA or RE is dispensable; otherwise, such rules are required
for the derivation.

\section{Amendments of Nute's Axiomatizations}

Nute's axiomatization for $\mathbf{V}$ and $\mathbf{VC}$ are as
follows\footnote{Nute's original axiomatization used the rule RCK instead of the axioms
CM and CC. But to reduce inference rules to the minimum, I prefer
to use these two axioms instead of the rule RCK. It can be easily
shown that they are equivalent as long as RCEC is provided.}
\[
\begin{aligned}\mathbf{Vn} & =\langle\mathrm{PC,ID,CM,CC,CV,MOD',CSO;RCEC}\rangle\\
\mathbf{VCn} & =\langle\mathrm{PC,ID,CM,CC,CV,MOD',CSO,CMP,CS;RCEC}\rangle.
\end{aligned}
\]
I will show that CA is not derivable in neither of these systems.
Since $\mathbf{VCn}$ is the stronger one, it suffices to prove that
CA is not derivable in $\mathbf{VCn}$. 
\begin{prop}
$\nvdash_{\mathbf{VCn}}\mathrm{CA}$.
\end{prop}
\begin{proof}
Let $U=\{0,1,2,3\}$, $A=\{1,2\}$, and $B=\{1,3\}$. Define $g:U\times\wp(U)\to\wp(U)$
as follows:
\[
\begin{aligned}g(i,X) & =\begin{cases}
\{1\} & \mbox{if }X=A\mbox{ and }i=0\\
\{i\} & \mbox{if }i\in X\\
X & \mbox{otherwise}
\end{cases}\end{aligned}
\]
Now I verify that $g$ satisfies the following conditions: for all
$i\in U$ and $X,Y\in\wp(U)$

\begin{lyxlist}{00.00.0000}
\item [{(id)}] $g(i,X)\subseteq X$
\item [{(mod)}] $g(i,X)=\emptyset\Longrightarrow g(i,Y)\cap X=\emptyset$
\item [{(cv)}] $g(i,X)\cap Y\neq\emptyset\Longrightarrow g(i,X\cap Y)\subseteq g(i,X)$
\item [{(cso)}] $g(i,X)\subseteq Y\mbox{ and }g(i,Y)\subseteq X\Longrightarrow g(i,X)=g(i,Y)$
\item [{(cent)}] $i\in X\Longrightarrow g(i,X)=\{i\}$
\end{lyxlist}
(id) and (cent) are obvious. (mod) holds since $g(i,X)=\emptyset$
iff $X=\emptyset$. It remains to verify (cv) and (cso). For (cv),
suppose $g(i,X)\cap Y\neq\emptyset$. Consider the following cases:
\begin{enumerate}
\item $X=A$ and $i=0$. Then $g(i,X)=\{1\}$. Since $g(i,X)\cap Y\neq\emptyset$,
we have $1\in Y$. Hence $X\cap Y=X$ or $X\cap Y=\{1\}$. In both
cases, we have $g(i,X\cap Y)=g(i,X)$.
\item $i\in X$. Then $g(i,X)=\{i\}$. Since $g(i,X)\cap Y\neq\emptyset$,
we have $i\in Y$. Then $i\in X\cap Y$. Hence $g(i,X\cap Y)=\{i\}=g(i,X)$.
\item $X\ne A$ or $i\neq0$, and $i\notin X$. Then $g(i,X)=X$. Since
$i\notin X$, we have $i\notin X\cap Y$. Then either $g(i,X\cap Y)=\{1\}$
or $g(i,X\cap Y)=X\cap Y$. If $g(i,X\cap Y)=X\cap Y$, we have $g(i,X\cap Y)\subseteq X=g(i,X)$.
If $g(i,X\cap Y)\neq X\cap Y$ and $g(i,X\cap Y)=\{1\}$, by the definition
of $g$, we have $X\cap Y=A$. Hence $1\in X$ and $g(i,X\cap Y)\subseteq g(i,X)$.
.
\end{enumerate}
For (cso), suppose $g(i,X)\subseteq Y$ and $g(i,Y)\subseteq X$.
Consider the following cases:
\begin{enumerate}
\item $X=A$ and $i=0$. Then $g(i,X)=\{1\}$. Since $g(i,X)\subseteq Y$,
we have $1\in Y$. Since $g(i,Y)\subseteq X$ and $i\notin X$, we
have $Y=A$ or $g(i,Y)=Y\subseteq X$. In the former case, we have
$g(i,X)=g(i,Y)$. In the latter case, we have $Y=\{1\}$, and hence
$g(i,Y)=\{1\}=g(i,X)$.
\item $i\in X$. Then $g(i,X)=\{i\}$. Since $g(i,X)\subseteq Y$, we have
$i\in Y$. Hence $g(i,Y)=\{i\}=g(i,X)$.
\item $X\neq A$ or $i\neq0$, and $i\notin X$. Then $g(i,X)=X$. Since
$g(i,X)\subseteq Y$, we have $X\subseteq Y$. If $Y=A$ and $i=0$,
then $g(i,Y)=\{1\}$. By $g(i,Y)\subseteq X$, we have $1\in X$.
Then by $X\subseteq Y=A$, we have $X=\{1\}$ or $X=Y$. In both cases
we have $g(i,X)=g(i,Y)$. If $i\in Y$, then $g(i,Y)=\{i\}$. Since
$g(i,Y)\subseteq X$, we have $i\in X$, contradicting that $i\notin X$.
In other cases, we have $g(i,Y)=Y$. Since $g(i,Y)\subseteq X$, we
have $Y\subseteq X$. Hence $g(i,X)=g(i,Y)$. 
\end{enumerate}
Given a model $\mathfrak{M}=(W,f,V)$, the truth set of $\varphi$
in $\mathfrak{M}$, denoted $[\varphi]^{\mathfrak{M}}$, is inductively
defined as follows:
\begin{itemize}
\item $[p]^{\mathfrak{M}}=V(p)$ for $p\in PV$
\item $[\neg\varphi]^{\mathfrak{M}}=W-[\varphi]^{\mathfrak{M}}$
\item $[\varphi\land\psi]^{\mathfrak{M}}=[\varphi]^{\mathfrak{M}}\cap[\psi]^{\mathfrak{M}}$
\item $[\varphi>\psi]^{\mathfrak{M}}=\{w\in W\mid f(w,[\varphi]^{\mathfrak{M}})\subseteq[\psi]^{\mathfrak{M}}\}$
\end{itemize}
\noindent{}We say that $\varphi$ is valid in $\mathfrak{F}=(W,f)$
if for all models $\mathfrak{M}$ based on $\mathfrak{F}$, $[\varphi]^{\mathfrak{M}}=W$.
Let $\mathfrak{G}=(U,g)$. By the frame conditions that $\mathfrak{G}$
satisfies, it can be easily verified that all axioms in $\mathbf{VCn}$
are valid in $\mathfrak{G}$, and $\mathfrak{G}$ preserves validity
for the rule RCEC. But CA is not valid in $\mathfrak{G}$, since $g(0,A\cup B)=\{1,2,3\}\nsubseteq\{1,3\}=g(0,A)\cup g(0,B)$.
Therefore $\nvdash_{\mathbf{Vn}}\mathrm{CA}$.

\end{proof}
\begin{cor}
$\nvdash_{\mathbf{Vn}}\mathrm{CA}$\footnote{Nute also gave the axiomatization $\mathbf{VWn}=\langle\mathrm{PC,ID,CV,MOD',CSO,CMP;RCEC},\mathrm{RCK}\rangle$
for Lewis' system $\mathbf{VW}$. So neither is $\mathrm{CA}$ derivable
from $\mathbf{VWn}$. The reason why $\mathrm{CA}$ is missing from
Nute's axiomatization is not clear, since no explicit proof of completeness
of these systems was given in his writings. I guess the reason may
be that $\mathrm{CA}$ is derivable in his axiomatization of $\mathbf{C2}$.}
\end{cor}
Now I show that by replacing $\mathrm{MOD}'$ with CA in the corresponding
systems, Nute's axiomatizations can be amended. Let
\[
\begin{aligned}\mathbf{Va} & =\langle\mathrm{PC,ID,CSO,DAE;RCK}\rangle\\
\mathbf{Vb} & =\langle\mathrm{PC,ID,MOD',PIE;RCK,RE}\rangle\\
\mathbf{Vc} & =\langle\mathrm{PC,ID,CM,CC,CV,CA,CSO;RCEC}\rangle\\
\mathbf{VCa} & =\langle\mathrm{PC,ID,CSO,DAE,CMP,CS;RCK}\rangle\\
\mathbf{VCb} & =\langle\mathrm{PC,ID,MOD',PIE,CMP,CS;RCK,RE}\rangle\\
\mathbf{VCc} & =\langle\mathrm{PC,ID,CM,CC,CV,CA,CSO,CMP,CS;RCEC}\rangle,
\end{aligned}
\]
where $\mathbf{Va}$ and $\mathbf{VCa}$ are Lewis' first axiomatizations
of $\mathbf{V}$ and $\mathbf{VC}$, respectively; $\mathbf{Vb}$
and $\mathbf{VCb}$ are his second axiomatizations; $\mathbf{Vc}$
and $\mathbf{VCc}$ are amendments of Nute's systems for $\mathbf{V}$
and $\mathbf{VC}$, respectively. To prove that these amendments are
equivalent to Lewis' original systems, it suffices to prove that $\mathbf{Vc}$
is equivalent to $\mathbf{Va}$ or $\mathbf{Vb}$. For completeness,
I will show that $\mathbf{Vc}$ is equivalent to both $\mathbf{Va}$
and $\mathbf{Vb}$.
\begin{prop}
\label{prop:Vc=00003DVa=00003DVb}$\mathbf{Vc}=\mathbf{Va}=\mathbf{Vb}$.
\end{prop}
\begin{proof}
First, I show that $\mathbf{Vc}\supseteq\mathbf{Va}$. For simplification
of proofs, I will show first that RCM, RCE, RCN, and RCEA are derivable
in $\mathbf{Vc}$.

For RCM:
\begin{enumerate}
\item $\varphi\to\psi$\hfill Assumption
\item $\varphi\land\psi\leftrightarrow\varphi$\hfill (1), PC
\item $(\chi>\varphi\land\psi)\leftrightarrow(\chi>\varphi)$\hfill (2),
RCEC
\item $(\chi>\varphi\land\psi)\to(\chi>\psi)$\hfill CM, PC
\item $(\chi>\varphi)\to(\chi>\psi)$\hfill (3), (4), PC
\end{enumerate}
For RCE:
\begin{enumerate}
\item $\varphi\to\psi$\hfill Assumption
\item $(\varphi>\varphi)\to(\varphi>\psi)$\hfill (1), RCM
\item $\varphi>\varphi$\hfill ID
\item $\varphi>\psi$\hfill (2), (3), PC
\end{enumerate}
For RCN:
\begin{enumerate}
\item $\psi$\hfill Assumption
\item $\varphi\to\psi$\hfill (1), PC
\item $\varphi>\psi$\hfill (2), RCE
\end{enumerate}
For RCEA:
\begin{enumerate}
\item $\varphi\leftrightarrow\psi$\hfill Assumption
\item $\varphi\to\psi$, $\psi\to\varphi$\hfill (1), PC
\item $\varphi>\psi$, $\psi>\varphi$\hfill (2), RCE
\item $(\varphi>\chi)\leftrightarrow(\psi>\chi)$\hfill (3), CSO, PC
\end{enumerate}
Now I prove that RCK is derivable in $\mathbf{Vc}$. The case for
$n=0$ is just RCN. The case for $n=1$ is just RCM. It remains to
prove the case for $n=2$. The case for $n>2$ can be obtained similarly.
\begin{enumerate}
\item $\psi_{1}\land\psi_{2}\to\psi$\hfill Assumption
\item $(\varphi>\psi_{1}\land\psi_{2})\to(\varphi>\psi)$\hfill (1), RCM
\item $(\varphi>\psi_{1})\land(\varphi>\psi_{2})\to(\varphi>\psi_{1}\land\psi_{2})$\hfill
CC
\item $(\varphi>\psi_{1})\land(\varphi>\psi_{2})\to(\varphi>\psi)$\hfill
(2), (3), PC
\end{enumerate}
Next I prove that DAE is derivable in $\mathbf{Vc}$. By CA, it suffices
to prove $(\varphi\lor\psi>\varphi)\lor(\varphi\lor\psi>\psi)\lor((\varphi\lor\psi>\chi)\to(\varphi>\chi)\land(\psi>\chi))$.
\begin{enumerate}
\item $(\varphi\lor\psi>\chi)\land\neg(\varphi\lor\psi>\neg(\neg\varphi\lor\psi))\to((\varphi\lor\psi)\land(\neg\varphi\lor\psi)>\chi)$\hfill
CV
\item $(\varphi\lor\psi)\land(\neg\varphi\lor\psi)\leftrightarrow\psi$\hfill
PC
\item $(\varphi\lor\psi>\chi)\land\neg(\varphi\lor\psi>\neg(\neg\varphi\lor\psi))\to(\psi>\chi)$\hfill
(1), (2), RCEA, PC
\item $(\varphi\lor\psi>\chi)\land\neg(\varphi\lor\psi>\neg(\varphi\lor\neg\psi))\to((\varphi\lor\psi)\land(\varphi\lor\neg\psi)>\chi)$\hfill
CV
\item $(\varphi\lor\psi)\land(\varphi\lor\neg\psi)\leftrightarrow\varphi$\hfill
PC
\item $(\varphi\lor\psi>\chi)\land\neg(\varphi\lor\psi>\neg(\varphi\lor\neg\psi))\to(\varphi>\chi)$\hfill
(4), (5), RCEA, PC
\item $\neg(\neg\varphi\lor\psi)\to\varphi$\hfill PC
\item $(\varphi\lor\psi>\neg(\neg\varphi\lor\psi))\to(\varphi\lor\psi>\varphi)$\hfill
(7), RCM
\item $\neg(\varphi\lor\psi>\varphi)\to\neg(\varphi\lor\psi>\neg(\neg\varphi\lor\psi))$\hfill
(8), PC
\item $\neg(\varphi\lor\neg\psi)\to\psi$\hfill PC
\item $(\varphi\lor\psi>\neg(\varphi\lor\neg\psi))\to(\varphi\lor\psi>\psi)$\hfill
(10), RCM
\item $\neg(\varphi\lor\psi>\psi)\to\neg(\varphi\lor\psi>\neg(\varphi\lor\neg\psi))$\hfill
(11), PC
\item $\neg(\varphi\lor\psi>\varphi)\land\neg(\varphi\lor\psi>\psi)\land(\varphi\lor\psi>\chi)\to(\varphi>\chi)\land(\psi>\chi)$\hfill
(3), (6), (9), (12), PC
\item $(\varphi\lor\psi>\varphi)\lor(\varphi\lor\psi>\psi)\lor((\varphi\lor\psi>\chi)\to(\varphi>\chi)\land(\psi>\chi))$\hfill
(13), PC
\end{enumerate}
This completes the proof of $\mathbf{Vc}\supseteq\mathbf{Va}$.

Now I prove $\mathbf{Va}\supseteq\mathbf{Vb}$. First I prove that
RCE, RCEA, and RCEC are derivable in $\mathbf{Va}$.

For RCE:
\begin{enumerate}
\item $\varphi\to\psi$\hfill Assumption
\item $(\varphi>\varphi)\to(\varphi>\psi)$\hfill (1), RCK
\item $\varphi>\varphi$\hfill ID
\item $\varphi>\psi$\hfill (2), (3), PC
\end{enumerate}
For RCEA:
\begin{enumerate}
\item $\varphi\leftrightarrow\psi$\hfill Assumption
\item $\varphi\to\psi$, $\psi\to\varphi$\hfill (1), PC
\item $(\varphi>\varphi)\to(\varphi>\psi)$, $(\psi>\psi)\to(\psi>\varphi)$\hfill
(2), RCK
\item $\varphi>\varphi$, $\psi>\psi$\hfill ID
\item $\varphi>\psi$, $\psi>\varphi$\hfill (3), (4), PC
\item $(\varphi>\chi)\leftrightarrow(\psi>\chi)$\hfill (5), CSO, PC
\end{enumerate}
For RCEC:
\begin{enumerate}
\item $\varphi\leftrightarrow\psi$\hfill Assumption
\item $\varphi\to\psi$, $\psi\to\varphi$\hfill (1), PC
\item $(\chi>\varphi)\to(\chi>\psi)$, $(\chi>\psi)\to(\chi>\varphi)$\hfill
(2), RCK
\item $(\chi>\varphi)\leftrightarrow(\chi>\psi)$\hfill (3), PC
\end{enumerate}
From RCEA and RCEC, RE is obtained by a simple induction.

Then I prove that CA is derivable in $\mathbf{Va}$. By DAE, it suffices
to prove $(\varphi\lor\psi>\varphi)\lor(\varphi\lor\psi>\psi)\to\mathrm{CA}$.\footnote{This is of course not a proper notation. It is used to abbreviate
the rather long formula (scheme) $(\varphi\lor\psi>\varphi)\lor(\varphi\lor\psi>\psi)\to((\varphi>\chi)\land(\psi>\chi)\to(\varphi\lor\psi>\chi))$.
I will use this kind of abbreviation occasionally below. The abbreviated
formulas should be easily recovered from context.}
\begin{enumerate}
\item $\varphi\to\varphi\lor\psi$, $\psi\to\varphi\lor\psi$\hfill PC
\item $\varphi>\varphi\lor\psi$, $\psi>\varphi\lor\psi$\hfill (1), RCE
\item $(\varphi\lor\psi>\varphi)\to((\varphi>\chi)\to(\varphi\lor\psi>\chi))$\hfill
(2), CSO, PC
\item $(\varphi\lor\psi>\psi)\to((\psi>\chi)\to(\varphi\lor\psi>\chi))$\hfill
(2), CSO, PC
\item $(\varphi\lor\psi>\varphi)\lor(\varphi\lor\psi>\psi)\to\mathrm{CA}$\hfill
(3), (4), PC
\end{enumerate}
Now I prove that MOD' is derivable in $\mathbf{Va}$.
\begin{enumerate}
\item $\varphi\land\neg\varphi\to\varphi\land\psi$\hfill PC
\item $(\neg\varphi>\varphi)\land(\neg\varphi>\neg\varphi)\to(\neg\varphi>\neg\varphi\land\psi)$\hfill
(1), RCK
\item $\neg\varphi>\neg\varphi$\hfill ID
\item $\neg\varphi\land\psi\to\neg\varphi$\hfill PC
\item $\neg\varphi\land\psi>\neg\varphi$\hfill (4), RCE
\item $(\neg\varphi>\varphi)\to(\neg\varphi\land\psi>\varphi)$\hfill (2),
(3), (5), CSO, PC
\item $\varphi\land\psi\to\varphi$\hfill PC
\item $\varphi\land\psi>\varphi$\hfill (7), RCE
\item $(\neg\varphi>\varphi)\to((\neg\varphi\land\psi)\lor(\varphi\land\psi)>\varphi)$\hfill
(6), (8), CA, PC
\item $(\neg\varphi\land\psi)\lor(\varphi\land\psi)\leftrightarrow\psi$\hfill
PC
\item $(\neg\varphi>\varphi)\to(\psi>\varphi)$\hfill (9), (10), RE
\end{enumerate}
Next I prove that PIE is derivable in $\mathbf{Va}$. It suffices
to prove
\begin{lyxlist}{00.00.0000}
\item [{(a)}] $(\varphi\land\psi>\chi)\to(\varphi>(\psi\to\chi))$, and
\item [{(b)}] $(\varphi>\neg\psi)\lor((\varphi>(\psi\to\chi))\to(\varphi\land\psi>\chi))$
\end{lyxlist}
For (a): Let $\alpha=(\varphi\land\psi)\lor(\varphi\land\neg\psi)$
\begin{enumerate}
\item $\varphi\land\neg\psi\to\neg\psi\lor\chi$\hfill PC
\item $\varphi\land\neg\psi>\neg\psi\lor\chi$\hfill (1), RCE
\item $\chi\to\neg\psi\lor\chi$\hfill PC
\item $(\varphi\land\psi>\chi)\to(\varphi\land\psi>\neg\psi\lor\chi)$\hfill
(3), RCK
\item $(\varphi\land\psi>\neg\psi\lor\chi)\land(\varphi\land\neg\psi>\neg\psi\lor\chi)\to(\alpha>\neg\psi\lor\chi)$\hfill
CA
\item $(\varphi\land\psi>\chi)\to(\alpha>\neg\psi\lor\chi)$\hfill (2),
(4), (5), PC
\item $\alpha\leftrightarrow\varphi,$ $\neg\psi\lor\chi\leftrightarrow(\psi\to\chi)$\hfill
PC
\item $(\varphi\land\psi>\chi)\to(\varphi>(\psi\to\chi))$\hfill (6), (7),
RE
\end{enumerate}
For (b): Let $\alpha=(\varphi\land\neg\psi)\lor(\varphi\land\psi)$
\begin{enumerate}
\item $(\alpha>\varphi\land\neg\psi)\lor(\alpha>\varphi\land\psi)\lor((\alpha>(\psi\to\chi))\to(\varphi\land\psi>(\psi\to\chi)))$\hfill
DAE, PC
\item $\alpha\leftrightarrow\varphi$\hfill PC
\item $(\varphi>\varphi\land\neg\psi)\lor(\varphi>\varphi\land\psi)\lor((\varphi>(\psi\to\chi))\to(\varphi\land\psi>(\psi\to\chi)))$\hfill
(1), (2), RE
\item $\psi\land(\psi\to\chi)\to\chi$\hfill PC
\item $(\varphi\land\psi>\psi)\land(\varphi\land\psi>(\psi\to\chi))\to(\varphi\land\psi>\chi)$\hfill
(4), RCK
\item $\varphi\land\psi>\psi$\hfill PC, RCE
\item $(\varphi>\varphi\land\neg\psi)\lor(\varphi>\varphi\land\psi)\lor((\varphi>(\psi\to\chi))\to(\varphi\land\psi>\chi))$\hfill
(3), (5), (6), PC
\item $(\varphi\land\psi)\land(\psi\to\chi)\to\chi$\hfill PC
\item $(\mbox{\ensuremath{\varphi}>\ensuremath{\varphi\land\psi})\ensuremath{\land}(\ensuremath{\varphi}>(\ensuremath{\psi\to\chi}))\ensuremath{\to}(\ensuremath{\varphi}>\ensuremath{\chi})}$\hfill
(8), RCK
\item $(\varphi>\varphi\land\psi)\land(\varphi>(\psi\to\chi))\to(\varphi\land\psi>\chi)$\hfill
(6), (9), CSO
\item $\varphi\land\neg\psi\to\neg\psi$\hfill PC
\item $(\varphi>\varphi\land\neg\psi)\to(\varphi>\neg\psi)$\hfill (11),
RCK
\item $\neg(\varphi>\neg\psi)\to\neg(\varphi>\varphi\land\neg\psi)$\hfill
(12), PC
\item $\neg(\varphi>\neg\psi)\land\neg(\varphi>\varphi\land\psi)\land(\varphi>(\psi\to\chi))\to(\varphi\land\psi>\chi)$\hfill
(7), (13), PC
\item $\neg(\varphi>\neg\psi)\land(\varphi>(\psi\to\chi))\to(\varphi\land\psi>\chi)$\hfill
(10), (14), PC
\item $(\varphi>\neg\psi)\lor((\varphi>(\psi\to\chi))\to(\varphi\land\psi>\chi))$\hfill
(15), PC
\end{enumerate}
This completes the proof of $\mathbf{Va}\supseteq\mathbf{Vb}$.

Now I prove $\mathbf{Vb}\supseteq\mathbf{Vc}$. The derivation of
CC and CM is straightforward using RCK. The rule RCEC is a special
case of RE. It remains to show that CA, CV, and CSO are derivable
in $\mathbf{Vb}$.

For CV:
\begin{enumerate}
\item $\neg(\varphi>\neg\psi)\land(\varphi>(\psi\to\chi))\to(\varphi\land\psi>\chi)$\hfill
PIE, PC
\item $\chi\to(\psi\to\chi)$\hfill PC
\item $(\varphi>\chi)\to(\varphi>(\psi\to\chi))$\hfill (2), RCK
\item $(\varphi>\chi)\land\neg(\varphi>\neg\psi)\to(\varphi\land\psi>\chi)$\hfill
(1), (3), PC
\end{enumerate}
For CSO: Let $\alpha=\varphi>\neg\psi$, $\beta=\psi>\neg\varphi$.
By PC, it suffices to prove
\begin{enumerate}
\item [(c)]$\neg\alpha\land\neg\beta\land(\varphi>\psi)\land(\psi>\varphi)\to((\varphi>\chi)\leftrightarrow(\psi>\chi))$,
\item [(d)]$\alpha\land(\varphi>\psi)\land(\psi>\varphi)\to((\varphi>\chi)\leftrightarrow(\psi>\chi))$,
and
\item [(e)]$\beta\land(\varphi>\psi)\land(\psi>\varphi)\to((\varphi>\chi)\leftrightarrow(\psi>\chi))$.
\end{enumerate}
For (c):
\begin{enumerate}
\item $\neg\alpha\to((\varphi\land\psi>\chi)\leftrightarrow(\varphi>(\psi\to\chi)))$\hfill
PIE
\item $\psi\land(\psi\to\chi)\to\chi$\hfill PC
\item $(\varphi>\psi)\land(\varphi>(\psi\to\chi))\to(\varphi>\chi)$\hfill
(2), RCK
\item $\chi\to(\psi\to\chi)$\hfill PC
\item $(\varphi>\chi)\to(\varphi>(\psi\to\chi))$\hfill (4), RCK
\item $\neg\alpha\land(\varphi>\psi)\to((\varphi>\chi)\leftrightarrow(\varphi\land\psi>\chi))$\hfill
(1), (3), (5), PC
\item $\neg\beta\land(\psi>\varphi)\to((\psi>\chi)\leftrightarrow(\varphi\land\psi>\chi))$\hfill
analogous to (1)--(6)
\item $\neg\alpha\land\neg\beta\land(\varphi>\psi)\land(\psi>\varphi)\to((\varphi>\chi)\leftrightarrow(\psi>\chi))$\hfill
(6), (7), PC
\end{enumerate}
For (d):
\begin{enumerate}
\item $\neg\psi\land\psi\to\chi$\hfill PC
\item $\alpha\land(\varphi>\psi)\to(\varphi>\chi)$\hfill (1), RCK
\item $\alpha\land(\varphi>\psi)\to(\varphi>\neg\varphi)$\hfill (2), PC
\item $(\varphi>\neg\varphi)\to(\psi>\neg\varphi)$\hfill MOD
\item $\alpha\land(\varphi>\psi)\to(\psi>\neg\varphi)$\hfill (3), (4),
PC
\item $\alpha\land(\varphi>\psi)\land(\psi>\varphi)\to(\psi>\varphi\land\neg\varphi)$\hfill
(5), RCK, PC
\item $\varphi\land\neg\varphi\to\chi$\hfill PC
\item $\alpha\land(\varphi>\psi)\land(\psi>\varphi)\to(\psi>\chi)$\hfill
(6), (7), RCK, PC
\item $\alpha\land(\varphi>\psi)\land(\psi>\varphi)\to((\varphi>\chi)\leftrightarrow(\psi>\chi))$\hfill
(2), (8), PC
\end{enumerate}
Note that in the above derivation, I use MOD instead of MOD', so that
I can dispense with RE. If MOD' is used instead, then the derivation
is longer, with an additional line of transforming MOD' to MOD, using
RE.

(e) can be proved analogously to (d).

For CA: Let $\alpha=\varphi\lor\psi>\neg\varphi$, $\beta=\varphi\lor\psi>\neg\psi$.
It suffices to prove
\begin{enumerate}
\item [(f)]$\neg\alpha\land\neg\beta\land(\varphi>\chi)\land(\psi>\chi)\to(\varphi\lor\psi>\chi)$,
\item [(g)]$\alpha\land(\varphi>\chi)\land(\psi>\chi)\to(\varphi\lor\psi>\chi)$,
and
\item [(h)]$\beta\land(\varphi>\chi)\land(\psi>\chi)\to(\varphi\lor\psi>\chi)$.
\end{enumerate}
For (f):
\begin{enumerate}
\item $\neg\alpha\land((\varphi\lor\psi)\land\varphi>\chi)\to(\varphi\lor\psi>(\varphi\to\chi))$\hfill
PIE, PC
\item $(\varphi\lor\psi)\land\varphi\leftrightarrow\varphi$\hfill PC
\item $\neg\alpha\land(\varphi>\chi)\to(\varphi\lor\psi>(\varphi\to\chi))$\hfill
(1), (2), RE
\item $\neg\beta\land(\psi>\chi)\to(\varphi\lor\psi>(\psi\to\chi))$\hfill
analogous to (1)--(3)
\item $\varphi\lor\psi>\varphi\lor\psi$\hfill ID
\item $(\varphi\lor\psi)\land(\varphi\to\chi)\land(\psi\to\chi)\to\chi$\hfill
PC
\item $(\varphi\lor\psi>(\varphi\to\chi))\land(\varphi\lor\chi>(\psi\to\chi))\to(\varphi\lor\psi>\chi)$\hfill
(5), (6), RCK, PC
\item $\neg\alpha\land\neg\beta\land(\varphi>\chi)\land(\psi>\chi)\to(\varphi\lor\psi>\chi)$\hfill
(3), (4), (7), PC
\end{enumerate}
For (g):
\begin{enumerate}
\item $(\varphi\lor\psi)\land\neg\varphi\to\psi$\hfill PC
\item $\varphi\lor\psi>\varphi\lor\psi$\hfill ID
\item $\alpha\to(\varphi\lor\psi>\psi)$\hfill (1), (2), RCK, PC
\item $\psi>\varphi\lor\psi$\hfill PC, RCK, ID
\item $\alpha\land(\psi>\chi)\to(\varphi\lor\psi>\chi)$\hfill (3), (4),
CSO
\item $\alpha\land(\varphi>\chi)\land(\psi>\chi)\to(\varphi\lor\psi>\chi)$\hfill
(5), PC
\end{enumerate}
(h) can be prove analogously to (g).

This completes the proof of $\mathbf{Vb}\supseteq\mathbf{Vc}$.

\end{proof}
\begin{cor}
\label{cor:VW-VC} $\mathbf{VCc}=\mathbf{VCa}=\mathbf{VCb}$.
\end{cor}

\section{\label{sec:A-New-Axiomatization}New Axiomatizations of Lewis' Conditional
Logics}

I propose the following new axiomatizations of Lewis' conditional
logics, which are denoted by $\mathbf{V}'$ and $\mathbf{VC}'$, respectively.
\[
\begin{aligned}\mathbf{V}' & =\langle\mathrm{PC,ID,CM,CA,CV,AC,RT;RCEC}\rangle\\
\mathbf{VC}' & =\langle\mathrm{PC,ID,CM,CA,CV,AC,RT,CMP,CS;RCEC}\rangle
\end{aligned}
\]

\noindent{}Both systems replace the axiom CSO by the axioms AC and
RT in $\mathbf{Vc}$ and $\mathbf{VCc}$, respectively. Meanwhile,
CC is omitted, since it is derivable from other axioms and rules.
I will prove that the new axiomatizations are equivalent to Lewis'
original ones. By Proposition~\ref{prop:Vc=00003DVa=00003DVb} and
Corollary~\ref{cor:VW-VC}, it suffices to prove that $\mathbf{V}'$
is equivalent to $\mathbf{Vc}$.
\begin{prop}
$\mathbf{V}'=\mathbf{Vc}$.
\end{prop}
\begin{proof}
First, I show that $\mathbf{Vc}\supseteq\mathbf{V}'$, i.e. AC and
RT are derivable in $\mathbf{Vc}$.

For AC:
\begin{enumerate}
\item $(\varphi>\varphi)\land(\varphi>\psi)\to(\varphi>\varphi\land\psi)$\hfill
CC
\item $\varphi>\varphi$\hfill ID
\item $(\varphi>\psi)\to(\varphi>\varphi\land\psi)$\hfill (1), (2), PC
\item $\varphi\land\psi\to\varphi$\hfill PC
\item $\varphi\land\psi>\varphi$\hfill (4), RCE
\item $(\varphi>\psi)\to(\varphi>\varphi\land\psi)\land(\varphi\land\psi>\varphi)$\hfill
(3), (5), PC
\item $(\varphi>\varphi\land\psi)\land(\varphi\land\psi>\varphi)\to((\varphi>\chi)\leftrightarrow(\varphi\land\psi>\chi))$\hfill
CSO
\item $(\varphi>\psi)\land(\varphi>\chi)\to(\varphi\land\psi>\chi)$\hfill
(6), (7), PC
\end{enumerate}
For RT:
\begin{enumerate}
\item $\psi\land\varphi\to\varphi$\hfill PC
\item $\psi\land\varphi>\varphi$\hfill (1), RCE
\item $\varphi>\varphi$\hfill ID
\item $(\varphi>\psi)\to(\varphi>\psi\land\varphi)$\hfill (3), CC, PC
\item $(\varphi>\psi)\to(\varphi>\psi\land\varphi)\land(\psi\land\varphi>\varphi)$\hfill
(2), (4), PC
\item $(\varphi>\psi\land\varphi)\land(\psi\land\varphi>\varphi)\to((\varphi>\chi)\leftrightarrow(\psi\land\varphi>\chi))$\hfill
CSO
\item $(\varphi>\psi)\land(\psi\land\varphi>\chi)\to(\varphi>\chi)$\hfill
(5), (6), PC
\end{enumerate}
Then I show that $\mathbf{V}'\supseteq\mathbf{Vc}$.

For CSO:
\begin{enumerate}
\item $(\varphi>\psi)\land(\varphi>\chi)\to(\varphi\land\psi>\chi)$\hfill
AC
\item $(\psi>\varphi)\land(\varphi\land\psi>\chi)\to(\psi>\chi)$\hfill
RT
\item $(\varphi>\psi)\land(\psi>\varphi)\land(\varphi>\chi)\to(\psi>\chi)$\hfill
(1), (2), PC
\item $(\varphi>\psi)\land(\psi>\varphi)\land(\psi>\chi)\to(\varphi>\chi)$\hfill
analogous to (1)--(3)
\item $(\varphi>\psi)\land(\psi>\varphi)\to((\varphi>\chi)\leftrightarrow(\psi>\chi))$\hfill
(3), (4), PC
\end{enumerate}
To prove CC, note that we have proved that RCE can be obtained from
PC, ID, CM, and RCEC in the proof of Proposition~\ref{prop:Vc=00003DVa=00003DVb}.
Since RCEA follows from RCE and CSO, we also have RCEA in $\mathbf{V}'$.
Now we have the following derivation for CC:
\begin{enumerate}
\item $\varphi\land\psi\land\chi>\varphi\land\psi\land\chi$\hfill ID
\item $\varphi\land\psi\land\chi>\psi\land\chi$\hfill (1), CM, PC
\item $\varphi\land\psi\land\chi\leftrightarrow\chi\land\varphi\land\psi$\hfill
PC
\item $\chi\land\varphi\land\psi>\psi\land\chi$\hfill (2), (3), RCEA,
PC
\item $(\varphi>\psi)\land(\varphi>\chi)\to(\varphi\land\psi>\chi)$\hfill
AC
\item $(\varphi>\psi)\land(\varphi>\chi)\to(\varphi\land\psi>\psi\land\chi)$\hfill
(4), (5), RT, PC
\item $\varphi\land\psi\leftrightarrow\psi\land\varphi$\hfill PC
\item $(\varphi\land\psi>\psi\land\chi)\to(\psi\land\varphi>\psi\land\chi)$\hfill
(7), RCEA
\item $(\varphi>\psi)\land(\psi\land\varphi>\psi\land\chi)\to(\varphi>\psi\land\chi)$\hfill
RT
\item $(\varphi>\psi)\land(\varphi>\chi)\to(\varphi>\psi\land\chi)$\hfill
(6), (8), (9), PC
\end{enumerate}
\end{proof}
\begin{cor}
 $\mathbf{VC}'=\mathbf{VCc}$.
\end{cor}
The axiom CSO was criticized by \citet{Gabbay1972}. One may be inclined
to abandon it directly. However, the above new systems show that CSO
can be recovered from AC and RT. It should be easy to notice that
AC and RT correspond to cautious monotonicity and cautious cut (a.k.a.
cumulative transitivity) in nonmonotonic logics. Both cautious monotonicity
and cautious cut are regarded as the minimal requirements for nonmonotonic
consequences. If AC and RT are also taken to be minimal for conditional
logics, then the above proof shows that CSO is inevitable in conditional
logics. If CSO is inevitable, then RCEA is also inevitable, since
it follows from CSO and another very modest axiom CM. The new axiomatization
indicates that it is difficult to construct nonclassical conditional
logics for characterizing default conditionals. It also leads us to
a puzzle about the controversial axiom SDA, which is the converse
of CA.

\section{A Resolution of a Puzzle about SDA}

The axiom SDA suggests that conditionals with disjunctive antecedents
have conjunctive reading. For example, when I say that if John or
Mary comes to my party, I'll be happy, it is reasonable to conclude
that if John comes to my party I'll be happy, and if Mary comes to
my party I'll be happy. But if SDA is contained in any conditional
logic with the rule RCEA, the so called fallacy of strengthening the
antecedent which is rejected in all conditional logics will be recovered.
This can be shown by the following simple derivation:
\begin{enumerate}
\item $\varphi\leftrightarrow(\varphi\lor(\varphi\land\psi))$\hfill PC
\item $(\varphi\lor(\varphi\land\psi)>\chi)\to\varphi\land\psi>\chi$\hfill
SDA, PC
\item $(\varphi>\chi)\to(\varphi\land\psi>\chi)$\hfill (1), (2), RCEA,
PC
\end{enumerate}
There are mainly three approaches to solving this puzzle. The first
approach, adopted in \citep{Loewer1976,McKay1977,Nute1980a,Lewis1977},
is to abandon SDA and apply something other than logic such as translation
lore to account for the intuitive validity of SDA. The second approach,
adopted in \citep{Nute1975,Nute1978,Nute1980}, is to keep SDA while
giving up the rule RCEA by developing nonclassical conditional logics.
As we have seen in Section~\ref{sec:A-New-Axiomatization}, this
means that some other intuitively reasonable axioms such as AC or
RT have to be abandoned too. In \citep{Fine1975b}, both the first
two approaches were suggested. The third approach, adopted in \citep{Alonso-Ovalle2006,Klinedinst2007,Paoli2012},
is to give nonclassical interpretations for disjunction, so that the
disjunctive antecedents in conditionals have conjunctive reading.
All the approaches are somewhat ad hoc, in the sense that conditionals
with disjunctive antecedents are treated as special and different
from other conditionals.

It has been noticed that SDA has counterexamples in both counterfactual
and indicative conditionals. The following is one for counterfactuals
given in \citep{McKay1977}:
\begin{enumerate}
\item If Spain fought on the Axis side or fought on the Allied side, it
would fight on the Axis side.
\item If Spain fought on the Allied side, it would fight on the Axis side.
\end{enumerate}
By SDA, (1) implies (2). But obviously (2) is false even if (1) is
true. A similar counterexample for indicative conditionals was given
in \citep{Carlstrom1978}:
\begin{enumerate}
\item [(3)]If Ivan is playing tennis or playing baseball, then he is playing
baseball.
\item [(4)]If Ivan is playing tennis, then he is playing baseball.
\end{enumerate}
By SDA, (3) implies (4). But we can have (3) true and (4) false. Both
counterexamples have the following form: $\varphi\lor\psi>\varphi$
is true but $\psi>\varphi$ false. As far as I know, no other forms
of counterexamples of SDA have been discovered. Considering that SDA
has only counterexamples of such special forms, one can not resist
keeping SDA while explaining away such counterexamples by attributing
them as abnormal uses of conditionals with disjunctive antecedents.
But we still face the conflict between SDA and RCEA. Remarkably, one
of Lewis' axioms for conditional logics, namely the old-fashioned
axiom DAE, which has been neglected for a long time, can perfectly
account for both the intuitive validity of SDA and its counterexamples!
The axiom DAE says that either $\varphi\lor\psi>\varphi$ is true,
or $\varphi\lor\psi>\psi$ is true, or $(\varphi\lor\psi>\chi)$ is
logically equivalent to $(\varphi>\chi)\land(\psi>\chi)$. From DAE
it follows that
\[
\neg(\varphi\lor\psi>\varphi)\land\neg(\varphi\lor\psi>\psi)\to\mathrm{SDA},
\]
which is weaker than SDA. But it is not too weak, since as long as
we exclude the cases when the disjunctive antecedent conditionally
implies one of its disjuncts, which are exactly the counterexamples
for SDA we have found, SDA is obtained. I think this resolution of
the puzzle around SDA is better than previous ones, since we can dispense
with any special treatments of the conditionals with disjunctive antecedents.
It is a big surprise that Lewis himself did not discover this simple
solution, even though he had published a note \citep{Lewis1977} about
SDA some years after he proposed the axiom DAE in \citep{Lewis1971}.


\end{document}